\documentclass[10pt]{llncs} 
\usepackage{amssymb,amsmath}
\usepackage{graphicx, epsfig}

\let\idit=\id

\def\idrm#1{\ensuremath{\mathrm{#1}}}

\newenvironment{itemize*}%
  {\begin{itemize}%
    \setlength{\itemsep}{0pt}%
    \setlength{\parskip}{0pt}%
    \setlength{\parsep}{0pt}%
    \setlength{\topsep}{0pt}%
    \setlength{\partopsep}{0pt}%
  }%
  {\end{itemize}}%

\newcommand{\pred}{\idrm{pred}}
\newcommand{\succes}{\idrm{succ}}

\newcommand{\cc}{{\idit count}}
\newcommand{\cM}{{\cal M}}
\newcommand{\cL}{{\cal L}}

\newcommand{\eps}{\varepsilon}

\renewenvironment{proof}{\trivlist\item[]\emph{Proof}:}%
{\unskip\nobreak\hskip 1em plus 1fil\nobreak$\Box$
\parfillskip=0pt%
\endtrivlist}

\newenvironment{proofsk}{\trivlist\item[]\emph{Proof Sketch}:}%
{\unskip\nobreak\hskip 1em plus 1fil\nobreak$\Box$
\parfillskip=0pt%
\endtrivlist}

\begin{document}

\title{Data Structures for Approximate  Orthogonal Range Counting}
\author{Yakov Nekrich}
\institute{Dept. of Computer Science\\ University of Bonn\\Email {\tt yasha@cs.uni-bonn.de}}
\date{Dept. of Computer Science\\ University of Bonn\\Email {\tt yasha@cs.uni-bonn.de}}
\maketitle
\begin{abstract}

We present new data structures for approximately counting the number of 
points in an orthogonal range. 
There is a  deterministic linear space data structure that 
supports updates in $O(1)$ time 
and approximates the number of elements in a 1-D range up to an additive term 
$k^{1/c}$ in $O(\log \log U\cdot\log \log  n)$ time, where $k$ is the number 
of elements in the answer, $U$ is the size of the universe
 and $c$ is an arbitrary fixed constant. 
We can estimate the number of points in a two-dimensional 
orthogonal range 
up to an additive term $ k^{\rho}$ in $O(\log \log U+ (1/\rho)\log\log n)$ 
time for any $\rho>0$. We can estimate the number of points in a 
 three-dimensional orthogonal range 
up to an additive term  $k^{\rho}$ in
$O(\log \log U + (\log\log n)^3+ (3^v)\log\log n)$ time for 
$v=\log \frac{1}{\rho}/\log \frac{3}{2}+2$.
\end{abstract}

\section{Introduction}
Range reporting and range counting are two variants of the range
searching problem.  In the range counting problem, the data structure
returns the number of points in an arbitrary query range. In the
range reporting problem the data structure reports all points in the
query range.  Both variants were studied extensively and in many cases
we know the matching upper and lower bounds for those problems for
dimension $d\leq 4$.  Answering an orthogonal range counting query takes more time
than answering the orthogonal range reporting query in the same dimension.
 This
gap cannot be closed because of the lower bounds for the range
counting queries: while range reporting queries can be answered in
constant time in one dimension and in almost-constant time in two and
three dimensions (if the universe size is not too big)\footnote{ For
  simplicity, we consider only emptiness queries. In other words,
  we ignore the time needed to output the points in the answer: if range
  reporting data structure supports queries in $O(f(n) +k)$ time, we
  simply say that the query time is $O(f(n))$.}, range counting
queries take super-constant time in one dimension and poly-logarithmic time
in two and three dimensions.  

\tolerance=1000
\emph{Approximate range counting queries} help us bridge the gap
 between range 
reporting and counting: instead of exactly counting
the number of points (elements) in the query range, the  data
structure provides a good estimation.  
There are data structures that approximate the number of points 
in a one-dimensional interval~\cite{ABR01,M06phd} or in a 
halfspace~\cite{AH08},~\cite{KS06},~\cite{AC07},~\cite{AHS07}  up to a 
constant factor: given a query $Q$, the data structure returns the number 
$k'$ such that $(1-\eps) k \leq k' \leq (1+\eps) k$, where $k$ is the exact 
number of points in the answer and $\eps$ is an arbitrarily small positive
 constant. 
In this paper  we consider the following new variant of approximate 
range counting:
If $k$ is the number of points in the answer, the answer to a query 
$Q$ is an integer $k'$ such that  
$k-\eps k^{\alpha} \leq k' \leq k+\eps k^{\alpha}$ 
for some constant $\alpha<1$. Thus we obtain better estimation for the
 number of 
points in the answer for large (superconstant) values of $k$. 
On the other hand, if the range $Q$ is empty, then $k'=0$.
We present data structures that approximate the number of points 
in a $d$-dimensional  orthogonal range  
for $d=2,3$. We also describe a dynamic one-dimensional data structure.
\\
{\bf Dynamic 1-D Data Structure.}  A static data structure that answers 1-D
reporting queries in $O(1)$ time is described in~\cite{ABR01}.
In~\cite{ABR01} the authors also describe a static data structure that
approximates the number of points in a 1-D range up to an arbitrary
constant factor in constant time. 
P\v{a}tra\c{s}cu and Demaine~\cite{PD06} show that any dynamic data structure
 with
polylogarithmic update time needs $\Omega(\log n/\log \log U)$ time to
answer an exact range counting query; henceforth $U$ denotes the size of the 
universe. 
The dynamic randomized data structure of Mortensen~\cite{M06phd} supports 
approximate range counting 
queries in $O(1)$ time and updates in $O(\log^{\eps} U)$ time; 
see~\cite{M06phd} for other trade-offs between query and update times. 
 In this paper we present a new result on
 approximate range counting in 1-D:
\begin{itemize*}
\item There is a deterministic data structure that can answer
  one-dimensional approximate range counting queries using the best known
  data
  structure for predecessor queries, i.e.  dynamic data
  structure supports range reporting queries in $O(dpred(n,U))$ time,
  where $dpred(n,U)$ is the time to answer a predecessor query in the
  dynamic setting; currently $dpred(n,U)=O(\min(\log \log U\cdot \log
  \log n, \sqrt{\log n/\log \log n}))$~\cite{AT07}. We show that we can
  approximate the number of points in the query range up to an
  additive factor $ k^{1/c}$, where $k$ is the number of points in the
  answer and $c$ is an arbitrary  constant, in $O(dpred(n,U))$
  time. We thus significantly improve the precision of the estimation; 
  the query time is still much less than the lower bound for the 
  exact counting queries in the dynamic scenario. 
\end{itemize*}
Using the standard techniques, we can extend the results for
one-dimensional approximate range counting to an arbitrary constant
dimension $d$.  There is a data structure that approximates the number
of points in a $d$-dimensional range up to an additive term $k^c$ for
any $c>0$ in $O(\log \log n (\log n/\log \log n)^{d-1} )$ time and
supports updates in $O(\log^{d-1+\eps} n)$ time.  For comparison, the fastest
 known dynamic data structure~\cite{M06} supports emptiness 
queries in $O((\log n/\log \log n)^{d-1})$ time. 
Dynamic data structures are described in section~\ref{sec:1d}.
\\{\bf Approximate Range Counting in 2-D and 3-D.}
We match or almost match the best upper bounds for 2-D and 3-D emptiness 
queries. Best data structures for exact range counting in 2-D and 3-D 
support queries in $O(\log n/\log \log n)$ and $O((\log n/\log \log n)^2)$ 
time respectively~\cite{JMS04}. 
\begin{itemize*}
\item If all point coordinates do not exceed $n$, we can approximate 
  the number of points in a two-dimensional query
  rectangle up to an additive term $ k^{\rho}$ for an
   arbitrary parameter $\rho$, $ 0 < \rho <1$, in
   $O(  (1/\rho)\log\log n)$ time. 
\item If all point coordinates do not exceed $n$, we can approximate 
  the number of points in three-dimensional query
  rectangle up to an additive term $ k^{\rho}$ in $O((\log \log n)^3 +
  (3^{v})\log\log n)$
  time for an arbitrary parameter $\rho$, $0 < \rho <1$, and 
  $v=\log \frac{1}{\rho}/\log \frac{3}{2}+2$. 
\end{itemize*}
The parameter $\rho$ is not fixed in advance, i.e. the same data structures 
can be used for answering queries with arbitrary precision. If point 
coordinates are 
arbitrary integers, then the query time of the above data structures increases 
by an additive term $O(\min (\log \log U, \sqrt{\log n/\log \log n}))$.
Data structure for range counting in 2-D and 3-D are described in section~\ref{sec:23d}. In section~\ref{sec:space} we describe space-efficient variants 
of two- and three-dimensional data structures that estimate the number 
of points in a range up to an additive error $k^c$ for some fixed constant 
$c$.

Our results for approximate range counting queries are valid in the word  RAM 
model. Throughout this paper $\eps$ denotes an arbitrarily small constant.

\section{Dynamic Approximate Range Counting}
\label{sec:1d}
\tolerance=1000
We  show that in the dynamic scenario 
answering one-dimensional counting queries with an additive error 
$k^{1/c}$ can be performed as efficiently as answering predecessor queries. 
The best known   deterministic data structure supports one-dimensional 
emptiness queries in $O(dpred(n, U))$ time, where 
$dpred(n,U)=\min(\sqrt{\log n /\log \log n}, \log \log U \cdot \log \log n)$ is the time needed to answer a predecessor query in dynamic scenario~\cite{A96},~\cite{AT07}. 
\begin{theorem}\label{theor:dyn1d2}
For any fixed constant $c>1$, there exists a linear space data structure that supports approximate range 
counting queries  with additive error $k^{1/c}$ 
in $O(dpred(n,U))$ time, deletions  in $O(\log \log n)$ amortized  time, and 
insertions in $O(dpred(n,U))$ amortized  time. 
\end{theorem}
\begin{proof}
First we observe that if the query interval contains less 
than $(\log \log n)^c$ points for an arbitrary constant $c$, 
$k=|P\cap [a,b]| \leq (\log \log n)^c$, 
then we can use a simple modification of the standard binary tree solution: 
the set $P$ is divided into groups of $(\log \log n)^c$ consecutive 
elements, i.e., $|G_i|=(\log \log n)^c$ and every element in $G_i$ 
is smaller than any element in $G_{i+1}$. Using a dynamic  data structure 
for predecessor queries we can find in $O((dpred(n,U))$ time 
the successor $a'$ of $a$ in $P$ and the predecessor $b'$  of $b$ in $P$. 
If $a$ and $b$ belong to the same group $G_i$, then we can count 
 elements in $[a,b]$ in $O(\log \log \log n)$ time using the 
standard binary range tree solution. 
If $a'$ and $b'$ belong to two consecutive groups $G_i$ and $G_{i+1}$, 
then we count the number of elements $e\in G_i$, $e\geq a$, 
and the number of elements $e'\in G_{i+1}$, $e'\leq b$. 
If $a'$ belongs to a group $G_i$ and $b'$ belongs to a group 
$G_j$ so that  $j>i+1$, then $[a,b]$ contains more than $(\log \log n)^c$ 
elements.  We also assume w.l.o.g. that $c>2$.


We maintain the exponential tree~\cite{A96},~\cite{AT07} 
for the set $P$. 
The root node has $\Theta(n^{1/c})$ children, so that each child 
node contains between $n^{(c-1)/c}/2$ and $2n^{(c-1)/c}$ points from $P$. 
In a general case, if a node $v$ contains $n_v$ points of $P$, then 
node $v$ has $\Theta(n_v^{1/c})$ children, so that each child contains 
between $n_v^{(c-1)/c}/2$ and $2n_v^{(c-1)/c}$ points from $P$.
The exponential tree can be maintained as described in~\cite{A96}, so that 
insertions and deletions  are supported in $O(\log \log n)$ time. 
Additionally in every node $v$ we store  the approximate number of 
elements in any consecutive sequence of children of $v$, denoted by $c_v(i,j)$: 
for any $i<j$, $n_{v_i}+n_{v_{i+1}}+\ldots+n_{v_j} -n_v^{3/c}/2\leq c_v(i,j)
\leq n_{v_i}+n_{v_{i+1}}+\ldots+n_{v_j} + n_v^{3/c}/2$. 
When $n_v^{3/c}/2$ elements are inserted into a node $v$ or deleted 
from $v$, we set $c_{v}(i,j)=n_{v_i}+n_{v_{i+1}}+\ldots+n_{v_j}$ for all $i<j$.
Recomputing $c_{v}(i,j)$ for a node $v$ takes $O(n_v^{2/c})$ time. 
Since insertion or deletion results in incrementing or decrementing  the 
value of $n_v$ in $O(\log \log n)$ nodes $v$,
recomputing $c_v(i,j)$  incurs an amortized cost $O(\log \log n)$.
Thus amortized cost of a delete operation is $O(\log \log n)$.
When we insert a new point, we also have to find its position in the 
exponential tree; therefore an insertion takes $O(dpred(n,U))$ time. 

We store $O(n_v^{2/c})$ auxiliary values in each node $v$; hence, we can 
show that the space usage is $O(n)$ in exactly the same way as 
in~\cite{A96,AT07}.

Given an interval $[a,b]$, we find $b'=\pred(b,P)$ and $a'=\succes(a,P)$
and identify the leaves of the exponential tree in which they are stored. 
The lowest common ancestor $q$ 
of those leaves can be found in $O(\log \log n)$ 
time because the height of the tree is $O(\log \log n)$.
If $a'$ and $b'$ are stored in the $i$-th and the $j$-th children of $q$ 
and $i+1<j$, then all elements stored in $q_{i+1},\ldots,q_{j-1}$ belong to 
$[a,b]$ and we initialize a variable $\cc$ to $c_v(i+1,j-1)$. Otherwise 
$\cc$ is set to $0$. Then, we traverse the path 
from $q$ to $a'$ and in every visited node $v$ we increment 
$\cc$ by  $c_v(i_v+1,r_v)$, such that $a'$ is in the $i_v$-th child of $v$, 
and $r_v$ is the total number of $v$'s children. 
Finally, we traverse the path from $q$ to $b'$ and in every visited node $v$ we increment 
$\cc$ by  $c_v(1,i_v-1)$, such that $b'$ is in the $i_v$-th child of $v$, 
Suppose that the variable $\cc$ was incremented by  $s_v>0$ when a node $v$
 was visited. Let $k_v$ be the exact number of elements in all children 
of $v$ whose ranges are entirely contained in $v$. Then,  
$k_v -n_v^{3/c} \leq s_v \leq k_v + n_v^{3/c}$.  Since 
$ k_v\ \geq n_v^{(c-1)/c}$,
 $k_v -k_v^{3/(c-1)} \leq s_v \leq k_v + k_v^{3/(c-1)}$.
Clearly, the total number of points equals to the sum of $k_v$ for all 
visited nodes $v$. The search procedure visits less than 
$c_h \log \log n$ nodes for a constant $c_h$. 
Hence, $k-k^{3/(c-1)}\log \log n \leq \cc \leq k+k^{3/(c-1)}\log \log n$ for $k=|P\cap [a,b]|$. 
Since $\log \log n \leq k^{1/(c-1)}$, 
$k-k^{4/(c-1)} \leq \cc \leq k+k^{4/(c-1)}$. 
We obtain the result of the Theorem by replacing $c$ with  $c'=\max(5c,5)$ 
in the above proof. 
\end{proof}
Our dynamic data structure can be extended to 
$d$ dimensions using the standard range tree~\cite{B80}. 
\begin{theorem}\label{theor:dynmd}
For any fixed constant $c>1$, there exists a data structure that supports $d$-dimensional 
 approximate range counting queries with additive error $k^{1/c}$  
in $O(\log\log n (\log n /\log \log n)^{d-1})$ time and updates in 
$O(\log^{d-1+\eps}n)$ amortized time.\\
\end{theorem}
\begin{proof}
This result can be obtained by combining the standard range tree technique 
(node degree in a range tree is $O(\log^{\eps'} n)$ for an appropriate 
constant $\eps'=\eps/(d-1)$) with the data structure for 
 one-dimensional approximate range counting of Theorem~\ref{theor:dyn1d2}.
Details will be given in the full version of this paper.
\end{proof}
\section{Approximate Range Counting in 2-D and 3-D}
\label{sec:23d}
 A point $p$ dominates a point $q$ if each coordinate of $p$ is greater than 
or equal to the corresponding coordinate of $q$. 
The goal of the (approximate) dominance counting query is to (approximately) count the number of points in $P$ that dominate $q$. 
The dominance query is equivalent to the orthogonal range query with a restriction that 
query range $Q$ is a product of half-open intervals. 
We start this section with a description of the data structure that estimates 
the number of points 
in the answer to a 2-D dominance query up to a constant factor. We can obtain 
a data structure for general orthogonal range counting queries using a standard technique. Then, we show that 
queries can be answered with higher precision without increasing the query time. Finally, we describe a data structure for approximate range counting in 3-D.
For simplicity, we only consider the case when all point coordinates 
are bounded by $n$. We can obtain the results for the case of arbitrarily 
large point coordinates by a standard reduction to rank space
 technique~\cite{GBT84}: the space usage remains linear and the query time 
increases by $pred(n,U)$ - the time needed to answer a static predecessor 
query. 
\begin{theorem}
\label{theor:domin2d1}
There exists a linear space data structure that answers approximate 
 two-dimensional 
dominance  range counting  queries on $n\times n$ grid in 
$O(\log \log n)$ time. 
\end{theorem}
A $t$-approximate boundary, introduced by Vengroff and Vitter~\cite{VV96} 
is a polyline $\cM$ 
consisting of $O(n/t)$ axis-parallel segments that partitions 
the space\footnote{In this section we assume that all 
points have positive coordinates}, so that  
every point $\cM$ is dominated by at most $2t$ and at least $t$ points 
of $P$. 
This notion can be straightforwardly extended to a $t_{\alpha}$-boundary
$\cM_{\alpha}$: $\cM_{\alpha}$ partitions the space into two parts,
and every point $\cM_{\alpha}$ is dominated by at most $\alpha\cdot t$ and 
at least $t$ points of $P$.
We can construct a $t_{\alpha}$-boundary with the same algorithm 
as in~\cite{VV96}. Let $p$ be a point with coordinates $(0,0)$. We move 
$p$ in the positive $x$ direction until $p$ is dominated by at most 
$\alpha t$ points. Then, we repeat the following steps until the 
$x$-coordinate of $p$ equals to $0$: a) move $p$ in 
$+y$ direction as long as $p$ is dominated by more than $t$ points of $P$
b) move $p$ in the $-x$ direction until $p$ is dominated by $\alpha t$ 
points of $P$.  
The path traced by $p$ is a $t_{\alpha}$-boundary; see Fig.~\ref{fig:bound2d} 
for an example. 
\emph{Inward corners} are formed when we move $p$ in $+y$ direction, 
i.e. inward corners mark the beginning of step a) resp.\ the end 
of step b). 
Inward corners of $\cM$ have a property  that no point of 
$\cM$   is strictly dominated by an inward corner and for every point 
$m\in \cM$ that 
is not an inward corner, there is an inward corner $m_i$ dominated by $m$.
There are $O(n/t)$ inward corners in a $t_{\alpha}$-approximate boundary 
because for every inward corner $c=(c_x,c_y)$ there are $(\alpha-1)t$ points 
that dominate $c$ and do not dominate inward corners whose
 $x$-coordinates are larger than $c_x$.
 \begin{figure}[tbp]
   \centering
   \includegraphics[width=70mm]{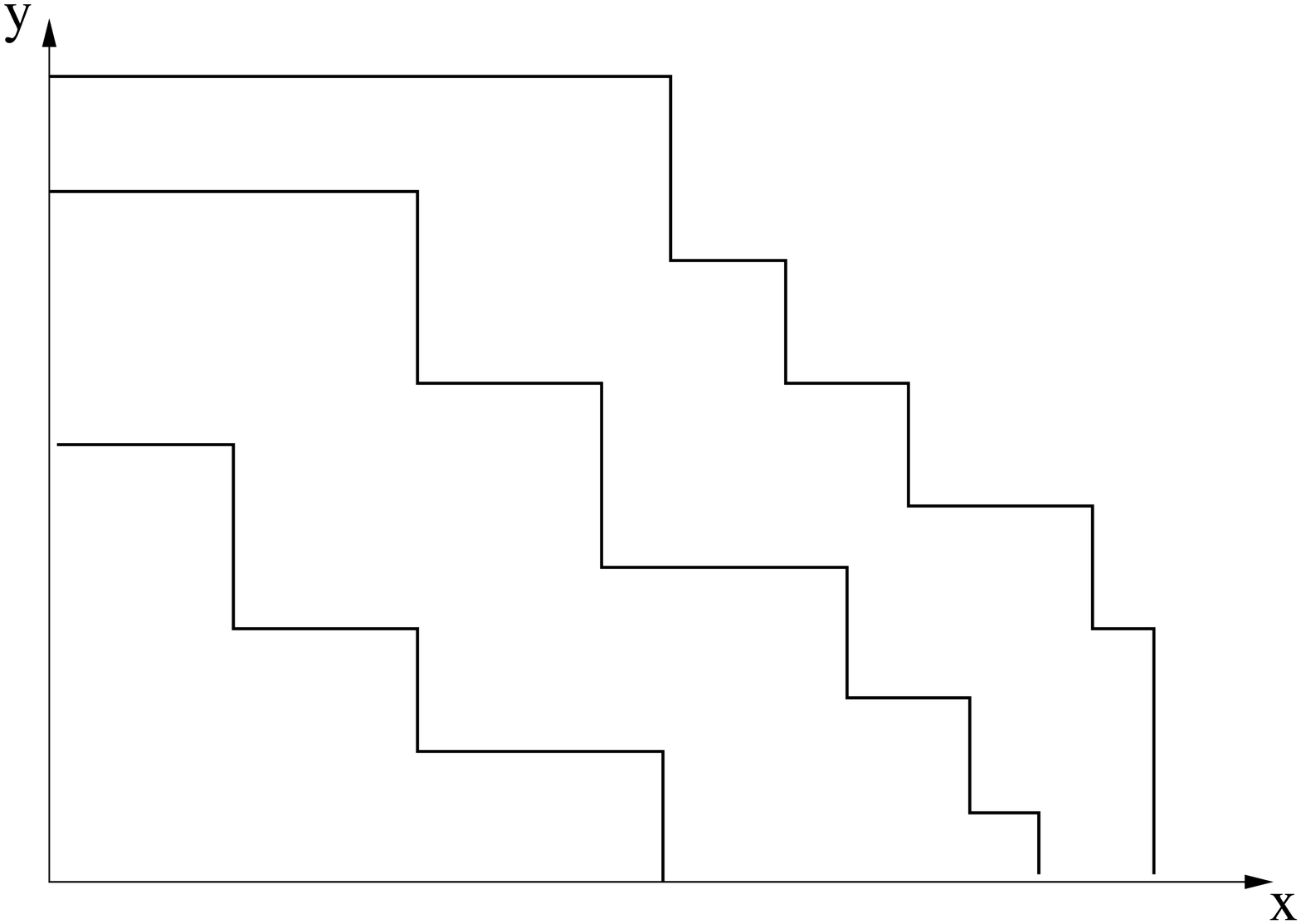}
   \caption{Example of $t$-approximate boundaries in 2-D. For simplicity, 
     the points of the set $P$ are not shown.}
   \label{fig:bound2d}
 \end{figure}
Our data structure consists of $\log_{\alpha} n$ $t_{\alpha}$-approximate 
boundaries $\cM_1,\cM_2,\ldots, \cM_s$ such that $\cM_i$ is an 
$\alpha^i$-approximate boundary of $P$, i.e. every point 
on $\cM_i$ is dominated by at least $\alpha^i$ and at most $\alpha^{i+1}$ 
points of $P$. If a point $p\in \cM_i$ is dominated by a query point $q$, 
then $q$ is dominated by at most $\alpha^{i+1}$ points of $P$. 
If $q$ dominates a point on $\cM_i$, then it also dominates an inward 
corner of $\cM_i$. Hence, we can estimate the number of points 
that dominate $q$ up to a constant $\alpha$ by finding the minimal index 
$j$ such that $q$ dominates an inward corner of $\cM_j$. 
Since $q$ is dominated by a point of $\cM_{j-1}$, $q$ is dominated by 
$k\geq \alpha^{j-1}$ points of $P$. On the other hand, $k\leq \alpha^{j+1}$ 
because a point of $\cM_j$ is dominated by $q$. 

We can store inward corners of all 
boundaries $\cM_i$ in a linear space data structure so that 
for any point $q$ the minimal index $j$, such that some point on $\cM_j$ 
is dominated 
by $q$, can be found in $O(\log \log n)$ time. 
 We denote by $\pred_x(a,S)$ the point 
$p=(p_x,p_y)\in S$, such that $p_x=\pred(a,S_x)$ where $S_x$ is the 
set of $x$-coordinates of all points in $S$.  
For simplicity, we sometimes do not distinguish between a boundary $\cM_i$ 
and the set of its inward corners.
 Let $q=(q_x,q_y)$.
Let $c_i=(c_x,c_y)$ be the inward corner on a boundary 
$\cM_i$ whose $x$-coordinate $c_x$ precedes $q_x$, $c_i=\pred_x(q_x,\cM_i)$.
For any other inward corner $c'_i=(c'_x,c'_y)$ on $\cM_i$,
$c'_y>c_y$ if and only if $c'_x<c_x$ because the $y$-coordinates 
of inward corners decrease monotonously as their $x$-coordinates increase.   
Hence, $q$ dominates a point on $\cM_i$ if and only if 
$q_y \geq c_y$. Thus given a query point $q$, it suffices to  identify 
the minimal index $j$, such that the $y$-coordinate of the inward 
corner $c_j\in \cM_j$ that precedes $q_x$ is smaller than or equal to $q_y$.
The $x$-axis is subdivided into intervals of size $\log n$. 
For each interval $I_s$ the list $L_s$ contains indexes of 
boundaries $\cM_i$ such that the $x$-coordinate of at least one inward 
corner of $\cM_i$ belongs to $I_s$. For a query point $q$ with $q_x\in I_s$
and for every $j\in L_s$, we can find the 
inward corner preceding $q_x$ with respect to its $x$-coordinate, $\pred_x(q_x,\cM_j)$, 
in $O(1)$ time because $x$-coordinates of all relevant inward corners belong to an 
interval of size $\log n$.   Hence, we can find the 
minimal index $j_s\in L_s$, such that $q$ dominates a point on $\cM_{j_s}$ 
in $O(\log \log n)$ time by binary search among indexes in $L_s$. 
For the left bound $a_s$ of an interval $I_s=[a_s,b_s]$ 
and for all indexes 
$j=1,\ldots, \log_{\alpha} n$, the list $A_s$ contains 
the inward corner $c_j$, such that $c_j=\pred_x(a_s,\cM_j)$.
By binary search in $A_s$ we can find 
the minimal $j_a$ such that $q$ dominates the inward 
corner $c_{j_a}\in A_s$. Clearly $j=\min(j_a,j_s)$ is the minimal 
index of a boundary dominated by $q$. 


\begin{theorem}\label{theor:2dgener}
There exists a $O(n\log^2 n)$ space data structure that supports 
two-dimensional approximate range counting queries on $n\times n$ grid 
in $O(\log \log n)$ 
time. 
\end{theorem}
The next Lemma will enable us to obtain a better estimation of the number 
of points. 
\begin{lemma}\label{lemma:adn}
There exists a $O(n\log n)$ space data structure that supports 
two-dimensional approximate range counting queries on $n\times n$ grid 
with an additive error $n^{\rho}$ in $O((1/\rho)\log\log n)$ 
time for any $\rho$, $0<\rho <1$. 
\end{lemma}
\begin{proof}
We divide the grid into $x$-slabs $X_i=[x_{i-1},x_i]\times [1,n]$ and 
$y$-slabs $Y_j=[1,n]\times [y_{j-1},y_j]$, so that each slab
contains $n^{1/2}$ points. 
For every point $(x_i,y_j)$, $0\leq i,j,\leq n^{1/2}$ we store the 
number of points in  $P$ that dominate it.  
There is also a recursively defined data structure for each slab. 
The total space usage is $s(n)=O(n)+2n^{1/2}s(n^{1/2})$ and 
$s(n)=O(n\log n)$. 

We can easily obtain an approximation with additive error $2n^{1/2}$ 
using the first level data structure: for a query $q=(q_x,q_y)$ we identify
the indexes $i$ and $j$, such that $x_{i-1} \leq q_x \leq x_i$ and 
$y_{j-1}\leq q_y \leq y_j$, i.e. we identify the $x$-slab $X_i$ and
 the $y$-slab $Y_j$ that contain $q$.  Indexes $i$ and $j$ can be found 
in $O(\log \log n)$ time.
Let $c(x,y)$ be the number of points that dominate a point $p=(x,y)$;
let $c(x,y,X_i)$ ($c(x,y,Y_j)$) be the number of points in 
the slab $X_i$ ($Y_j$) that dominate $p=(x,y)$. 
Then $c(q_x,q_y)=c(x_i,y_j) + c(x_i,q_y,Y_j) + c(q_x,q_y,X_i)$. 
Since  $c(x_i,q_y,Y_j)\leq n^{1/2}$ and $c(q_x,q_y,X_i)\leq n^{1/2}$, 
the value of $c(x_i,y_j)$ is an approximation of 
$c(q_x,q_y)$ with an additive error $2n^{1/2}$.  
Using recursive data structures for slabs $X_i$ and $Y_j$ we can estimate 
$c(q_x,q_y,X_i)$ and $ c(x_i,q_y, Y_j)$ with an additive error $2n^{1/4}$
and estimate $c(q_x,q_y)$ with an additive error $4n^{1/4}$. 
If the recursion depth is $v$ (i.e. if we apply recursion $v$ times),
then the total number of recursive calls is $O(2^v)$ and 
 we  obtain in $O((2^v)\log\log n)$ time  an approximation 
with additive error  $2^v\cdot n^{1/2^v}$ for any positive integer $v$.

We set recursion depth $v=\lceil \log (1/ \rho) \rceil +2$.
Then, $v  + (1/2^v) \log n \leq (\rho/4)\log n + \log (1/\rho)=
(\rho/4 +\frac{\log (1/\rho)}{\log n})\log n < \rho\log n$. 
Hence, $n^{\rho} > 2^v n^{1/2^v}$. 
Therefore, if recursion depth is set to $v$, 
then our data structure provides an answer with additive error $n^{\rho}$.
\end{proof}

\begin{theorem}\label{theor:domin2d2}
There exists a $O(n\log^2 n)$ space data structure that supports 
two-dimensional dominance  counting queries on $n\times n$ grid 
with an additive error 
$k^{\rho}$ for an arbitrary parameter $\rho$, $0<\rho <1$, in 
$O( (1/\rho)\log \log n)$ time.\\
There exists a $O(n\log^4 n)$ space data structure that supports 
two-dimensional  range counting queries  on $n\times n$ grid with an additive
 error 
$k^{\rho}$ for an arbitrary parameter $\rho$, $0<\rho <1$,
 in $O( (1/\rho)\log \log n)$ time. 
\end{theorem}
\begin{proof}
As in Theorem~\ref{theor:domin2d1} we construct $t$-boundaries 
$\cM_1, \ldots, \cM_{\log n}$, such that $M_i$ is a $2^i$-approximate 
boundary, i.e. each point on $\cM_i$ is dominated by  at 
least $2^i$ and at most $2^{2i}$ points of $P$. For each 
inward corner $c_{i,j}$ of every $M_j$, we store a data structure 
$D_{i,j}$ that contains all points that dominate $c_{i,j}$ and 
supports approximate counting queries as described in 
Lemma~\ref{lemma:adn}. For a fixed $j$, there are $O(\frac{n}{2^j})$ 
 data structures $D_{i,j}$, and each $D_{i,j}$ contains $O(2^j)$ points.
Hence, all data structures $D_{i,j}$ use $O(n\log^2 n)$ space.

As described in Theorem~\ref{theor:2dgener}, we can find in 
$O(\log \log n)$ time the minimal index $j$, such that 
$\cM_j$ is dominated by the query point $q$ and an inward corner 
$c_{i,j}\in \cM_j$ dominated by $q$. 
Then, we use the data structure $D_{i,j}$ to obtain a better approximation. 
Since $D_{i,j}$ contains $O(k)$ points, by Lemma~\ref{lemma:adn} 
$D_{i,j}$  estimates the number of points that dominate $q$ with an 
additive error $k^{\rho}$ in $O((1/\rho)\log\log n)$ time.  
We can extend the result for dominance counting to the general 
three-dimensional counting using the standard technique from range 
reporting~\cite{CG86b,SR95}; see also the proof of Theorem~\ref{theor:2dgener}.
\end{proof}

\begin{lemma}\label{lemma:adn2}
There exists a $O(n\log^3 n)$ space data structure that supports 
three-dimensional approximate range counting queries on $n\times n\times n$ grid 
with an additive error $n^{\rho}$ in $O(3^{v} \log \log n)$ time for any
 $\rho$,
$0<\rho <1$, and 
for $v=\log \frac{1}{\rho}/\log \frac{3}{2}+2$. 
\end{lemma}
\begin{proof}
We divide the grid into $x$-, $y$-, and $z$-slabs, $X_i=[x_{i-1},x_i]\times [1,n]\times [1,n]$, $Y_j=[1,n]\times [y_{j-1},y_j]\times [1,n]$, $Z_d=[1,n]\times [1,n]\times [z_{d-1},z_d]$, so that each slab
contains $n^{2/3}$ points. For each point $(x_i,y_j,z_d)$ we store the number of points in $P$ that dominate it. There is also a recursively defined data structure
 for each slab. 
The total space usage is $s(n)=O(n)+3n^{1/3}s(n^{2/3})$ and 
$s(n)=O(n\log^3 n)$.

For a query $q=(q_x,q_y,q_z)$ we identify the $x$-, $y$-, and $z$-slabs $X_i$, $Y_j$, and $Z_d$ that contain $q$. By the same argument as in Lemma~\ref{lemma:adn}, 
the number of points 
that dominate $(x_i,y_j,z_d)$ differs from the number of points that dominate 
$q$ by at most $3n^{2/3}$. 
We can estimate the number of points that dominate $q$ and belong to one 
of the slabs $X_i$, $Y_j$, and $Z_d$ using recursively defined data structures. 
If the recursion depth is $v$,
then we  obtain in $O(3^v\log \log n)$ time  an approximation 
with additive error  $3^v\cdot n^{(2/3)^v}$ for any positive integer $v$.
The result of the Lemma follows if we set  
$v=\log \frac{1}{\rho}/\log \frac{3}{2}+2$.
\end{proof}

\begin{theorem}\label{theor:approx3d}
There exists a $O(n\log^4 n)$ space data structure that supports 
 approximate dominance range counting queries on $n\times n\times n$ grid with an
 additive
 error $k^{\rho}$in $O((\log \log n)^3+ 3^v \log \log n)$ time for  any
 $\rho$, $0<\rho <1$,  and 
for $v=\log \frac{1}{\rho}/\log \frac{3}{2}+2$.\\
There exists a  $O(n\log^7 n)$ space data structure that supports 
 approximate range counting queries on $n\times n\times n$ grid  with
 an additive error $k^{\rho}$ in 
$O((\log \log n)^3+ 3^v \log \log n)$ time for  any $\rho$, $0<\rho <1$, and 
for $v=\log \frac{1}{\rho}/\log \frac{3}{2}+2$.
\end{theorem}
\begin{proof}
Instead of counting points that dominate $q$ we count points dominated 
by $q$. Both types of queries are equivalent. Hence, the data structure 
of Lemma~\ref{lemma:adn2} can be used to approximately count points dominated 
by $q$. 

A downward corner of a point $p$ consists of all points dominated by $p$.
We define an approximate $t$-level as a set of downward corners $\cL$, 
such that 
(1) any point $p$ that dominates at most $t$ points of $P$ is contained 
in some $r\in \cL$ (2) any downward corner $r\in \cL$ contains at most 
$\alpha\cdot t$ points of $P$.
Afshani~\cite{A08} showed that for an arbitrary constant $\alpha$ there 
exists an approximate $t$-level of size $O(\frac{n}{t})$. 
We can assume  that no $r\in \cL$ dominates $r'\in \cL$ in an 
approximate $t$-level $\cL$: 
if $r$ dominates $r'$, then the downward corner $r'$ can be removed from 
$\cL$. Identifying an inward corner $r\in \cL$ that dominates a query point 
$q$ (or answering that no $r\in \cL$ dominates $q$) is equivalent to 
answering a point location query in a rectangular planar 
subdivision~\cite{VV96,N07} and takes $O((\log\log n)^2)$ time. 
 
Our data structure consists of approximate levels 
$\cM_1,\cM_2,\ldots,\cM_{\log n}$, such that $\cM_i$ is a $2^i$-approximate 
level and the constant $\alpha$ is chosen to be $2$. For every 
downward corner $r_{i,j}\in \cM_j$, we store all points dominated by 
$r_{i,j}$ in a data structure $D_{i,j}$; $D_{i,j}$ contains $O(2^j)$ points 
and supports counting queries with additive error $O(2^{\rho j})$ by 
Lemma~\ref{lemma:adn2}. All data structures $D_{i,j}$ use $O(n\log^4 n)$ space.

We can find a minimal $j$, such that $\cM_j$ dominates 
$q$ in $O((\log \log n)^3)$ time by binary search.  Let $r_{i,j}$ be the 
downward corner that dominates 
$q$. We can use the data structure $D_{i,j}$ to estimate the number of 
points that are dominated by $q$ with an additive error  $k^{\rho}$; 
by Lemma~\ref{lemma:adn2} this takes $O(3^{v} \log\log n)$ time for 
 $v=\log \frac{1}{\rho}/\log \frac{3}{2}+2$. 

We can extend the result for dominance counting to the general 
three-dimensional counting using the standard technique~\cite{CG86b,SR95}; see
 also the proof of Theorem~\ref{theor:2dgener}. 
\end{proof}
\subsection{Space-Efficient Approximate Range Counting in 2-D and 3-D }
\label{sec:space}
If we are interested in counting with an 
additive error $k^{c}$ for some predefined constant $c>0$, then 
the space usage can be significantly reduced. The two-dimensional 
data structure uses $O(n\log^2 n)$ space ($O(n)$ space for dominance 
counting), and the three-dimensional data structure uses $O(n\log^3 n)$ 
space ($O(n)$ space for dominance counting). The main idea of our 
improvement is that 
in the construction of Lemma~\ref{lemma:adn} (resp.\ Lemma~\ref{lemma:adn2}) 
each 
slab contains $n^{1/2+\eps}$ points ($n^{2/3+\eps}$ points) for some $\eps>0$ 
 and there is a constant 
number of recursion levels. 
\begin{lemma}\label{lemma:adn3}
For any fixed constant $c<1$, there exists a $O(n^{1-\eps})$ 
space data structure that supports 
two-dimensional approximate range counting queries on $n\times n$ grid 
with an additive error $n^{c}$ in $O(\log\log n)$ 
time. 
\end{lemma}
 \begin{proof}
 We divide the grid into $x$-slabs $X_i=[x_{i-1},x_i]\times [1,n]$ and 
 $y$-slabs $Y_j=[1,n]\times [y_{j-1},y_j]$, so that each slab
 contains $n^{1/2+\eps}$ points. As in Lemma~\ref{lemma:adn}, we store 
 for each point $(x_i,y_j)$, $0\leq i,j,\leq n^{1/2-\eps}$, the number 
 of points in $P$ that dominate it. 
 Note that there are $O(n^{1-2\eps})$ points $(x_i,y_j)$ for 
 $0\leq i,j\leq n^{1/2-\eps}$ . 
 If an $x$-slab or a $y$-slab contains more than $n^{f}$ points for 
 a constant $f=c/4$, we store a recursively defined data
  structure for that slab. The number of recursion levels is 
 $g= \lceil \frac{\log (1/f)}{\log (2/(1+2\eps))}\rceil$. 
 Since each point is stored in one recursively defined data structure for 
 an $x$-slab and in one recursively defined data structure for a $y$-slab, 
 the total number of points in all recursively defined data  structures 
 increases by factor $2$ with each recursion level.  
 Thus the total space usage is $\sum_{k=1}^{g} 2^g\cdot O(n^{1-\eps})=O(n^{1-\eps})$.

 Given a query $q=(q_x,q_y)$, we identify the $x$-slab $X_i$ and
  the $y$-slab $Y_j$ that contain $q$. 
 Let $c(x,y)$ be the number of points that dominate a point $p=(x,y)$;
 let $c(x,y,X_i)$ ($c(x,y,Y_j)$) be the number of points in 
 the slab $X_i$ ($Y_j$) that dominate $p=(x,y)$.
 As in the proof of Lemma~\ref{lemma:adn},
  $c(q_x,q_y)=c(x_i,y_j) + c(x_i,q_y,Y_j) + c(q_x,q_y,X_i)$, where $X_i$ and 
 $Y_j$ are the $x$-slab and the $y$-slab that contain $q$.
 If slabs $X_i$ and $Y_j$, contain more than $n^f$ points, we estimate 
 $c(x_i,q_y,Y_j)$ and  $c(q_x,q_y,X_i)$  using data structures for slabs
 $Y_j$ and $X_i$. 
 Otherwise we use $c(x_i,y_j)$ as an estimation for $c(q_x,q_y)$. 
 By the same argument as in the proof of Lemma~\ref{lemma:adn}, we obtain 
 an approximation with additive error $2^g\cdot n^{f}$. Since $g < 2\log(1/f)$
 and $f = c/4$, $g+f\log n < 2\log(1/f) + (c/4) \log n < c \log n$. 
 Hence, $2^g\cdot n^{f} < n^c$ and we estimate the number of points in a 
 range with an additive error that is less than $n^c$.
 \end{proof}
Using Lemma~\ref{lemma:adn3}, we can prove the following Theorem.
\begin{theorem}\label{theor:2dlin}
For any fixed constant $c<1$, there exists a $O(n)$ space data structure that 
supports two-dimensional dominance  counting queries on $n\times n$ grid 
with an additive error 
$k^{c}$ in 
$O(\log \log n)$ time.\\
For any fixed constant $c<1$, there exists a $O(n\log^2 n)$ space data 
structure that supports two-dimensional  range counting queries  on 
$n\times n$ grid with an additive  error 
$k^{c}$ in $O(\log \log n)$ time. 
\end{theorem}
\begin{proof}
We construct a sequence of $t$-approximate boundaries $\cM_i$ in the same 
way as in Theorem~\ref{theor:domin2d2} and store all points that 
dominate an inward corner $c_{i,j}$ in data structure $D_{i,j}$. 
The only difference is that $D_{i,j}$ is implemented as described in Lemma~\ref{lemma:adn3}.  For a fixed $j$, there are $O(\frac{n}{2^j})$ 
 data structures $D_{i,j}$, and each $D_{i,j}$ needs $O(2^{(1-\eps)\cdot j})$ space.
Hence, all data structures $D_{i,j}$ use $O(\sum_j \frac{n}{2^{\eps\cdot j}})=O(n)$
space. 

Dominance queries are processed in exactly the same way as in Theorem~\ref{theor:domin2d2}. We can extend the result for dominance counting to the general 
two-dimensional counting using the standard technique from range 
reporting~\cite{CG86b,SR95}; see also the proof of Theorem~\ref{theor:2dgener}.
\end{proof}

\begin{lemma}\label{lemma:adn4}
For any fixed constant $c<1$, there exists a $O(n^{1-\eps})$ space data
structure that supports three-dimensional approximate range counting queries
on $n\times n \times n$ grid with an additive error $n^{c}$ in 
$O(\log\log n)$  time. 
\end{lemma}
 \begin{proofsk}
 Like in Lemma~\ref{lemma:adn2}, we divide the grid into 
 $x$-, $y$-, and $z$-slabs, $X_i=[x_{i-1},x_i]\times [1,n]\times [1,n]$, $Y_j=[1,n]\times [y_{j-1},y_j]\times [1,n]$, $Z_d=[1,n]\times [1,n]\times [z_{d-1},z_d]$, but  each slab
 contains $n^{2/3+\eps}$ points. 
 For each point $(x_i,y_j,z_d)$ we store the number of points in $P$ that dominate it. If the number of points in a slab is greater than $n^f$ for $f=c/16$, 
 then   we  store a recursively defined data structure for each slab. 

 We can estimate the space usage and analyze the query algorithm in the same 
 way as in Lemma~\ref{lemma:adn3}.
 \end{proofsk}
\begin{theorem}\label{theor:approx3d2}
For any fixed constant $c<1$, there exists a $O(n)$ space data structure 
that supports approximate dominance range counting queries on 
$n\times n\times n$ grid with an additive error $k^{c}$ in $O((\log \log n)^3)$ 
time.\\
For any fixed constant $c<1$, there exists a  $O(n\log^4 n)$ space data 
structure that supports  approximate range counting queries on 
$n\times n\times n$ grid  with an additive error $k^{c}$ in 
$O((\log \log n)^3)$ time.
\end{theorem}
 \begin{proofsk}
 As in the proof of Theorem~\ref{theor:approx3d} our data structure 
 consists of $2^i$-approximate levels $\cM_i$ for $i=1,\ldots, \log n$. 
 For every inward corner $r_{i,j}\in \cM_j$, we store all points dominated 
 by $r_{i,j}$ in the data structure $D_{i,j}$  described 
 in Lemma~\ref{lemma:adn4}.  
 Each $D_{i,j}$ uses $O(2^{(1-\eps)j})$ space. Since a $2^j$-approximate level 
  $\cM_{j}$ has 
 $O(\frac{n}{2^j})$  inward corners, all $\cM_j$ use 
 $O(\sum_j \frac{n}{2^{\eps\cdot j}})=O(n)$ space. 

 Dominance counting queries are answered in the same way as in 
 Theorem~\ref{theor:approx3d}. 
  We can extend the result for dominance counting to the general 
 three-dimensional counting by applying the standard technique from range 
 reporting~\cite{CG86b,SR95} that was  also used in proofs of Theorems~\ref{theor:2dgener}, \ref{theor:approx3d}, \ref{theor:2dlin}.
\end{proofsk}

\section*{Acknowledgment}
We would like to thank an anonymous reviewer of the previous version 
of this paper for stimulating  suggestions 
that helped us  improve some  of our results.

 \section*{Appendix A. Proof of Theorem~\ref{theor:2dgener} }
 We use the well known technique  used for range reporting queries~\cite{CG86b,SR95}. 
 The set of points $P$ is  subdivided into subsets $P_1, 
 P_2,\ldots, P_s$,  
 so that the total number of points in 
 $P_1\cup\ldots\cup P_s$ is $O(n\log^2 n)$, and an arbitrary query rectangle 
 $Q$ 
 can be represented as a union of  at most four rectangles $Q_1,\ldots,Q_s$, $s\leq 4$,
 so that $Q\cap P= (Q_1\cap P_{i_1})\cup\ldots\cup (Q_s\cap P_{i_s})$ 
 and each $Q_i$ is a product of two half-open intervals.  
 We store the date structure for approximate dominance queries of 
 Theorem~\ref{theor:domin2d1} for each set 
 $P_i$, so that the total space usage is $O(n\log^2 n)$. 
 Given a query $Q$, we can decompose $Q$ into $Q_1,\ldots, Q_s$ 
 and find the corresponding $P_{i_1},\ldots,P_{i_s}$ in $O(\log \log n)$ 
 time, see e.g.~\cite{N07}. Then, we can estimate the number of points in 
 each $P_{i_j}\cap Q_j$, $1\leq j\leq s$, and thus estimate the number of points in 
 $Q\cap P= (Q_1\cap P_{i_1})\cup\ldots\cup (Q_s\cap P_{i_s})$


\begin{thebibliography}{99}
\bibitem{A08} 
P. Afshani {\em On Dominance Reporting in 3D},
Proc.  ESA 2008, 41-51.
\bibitem{AC07}
P. Afshani, T. M. Chan, {\em 
On Approximate Range Counting and Depth}, Proc.  SoCG 2007, 337-343.
\bibitem{ABR00} S. Alstrup, G. S. Brodal, T. Rauhe 
{\em New Data Structures for Orthogonal Range Searching}, 
 Proc.  FOCS, 198-207,  2000.
\bibitem{ABR01}
S. Alstrup, G. S. Brodal, T. Rauhe,  {\em 
Optimal Static Range Reporting in One Dimension}, Proc. STOC 2001, 476-482.
\bibitem{A96}
A. Andersson, {\em Faster Deterministic Sorting and Searching in Linear Space} , Proc. FOCS 1996, 135-141.
\bibitem{AT07}
A. Andersson, M. Thorup, {\em  Dynamic Ordered Sets with Exponential Search 
Trees}  J. ACM (JACM) 54(3):13 (2007).
\bibitem{AH08}
B. Aronov, S. Har-Peled, {\em On Approximating the Depth and Related Problems},
 SIAM J. Comput. 38(3): 899-921 (2008).
\bibitem{AHS07}
B. Aronov, S. Har-Peled, M. Sharir, {\em On Approximate Halfspace Range Counting and Relative Epsilon-Approximations},  Proc. SoCG 2007, 327-336.
\bibitem{BF02} 
P. Beame, F. E. Fich, \newblock {\em Optimal
    Bounds for the Predecessor Problem and Related Problems}, J.
  Comput. Syst. Sci. 65(1): 38-72 (2002).
\bibitem{B80} 
J. L. Bentley,
\newblock  {\em Multidimensional Divide-and-Conquer},
\newblock  Commun. ACM 23: 214-229, 1980.
\bibitem{BKS95}
M. de Berg, M.~ J. van Kreveld, J. Snoeyink,  {\em Two- and Three-Dimensional
 Point Location in Rectangular Subdivisions}, 
J. Algorithms 18(2): 256-277 (1995).
\bibitem{CG86b}
B. Chazelle, L.~J. Guibas, {\em Fractional Cascading: I. A Data Structuring Technique}, Algorithmica 1(2): 133-162 (1986).
\bibitem{GBT84}
H. Gabow, J.~L. Bentley, R.~E. Tarjan,
{\em Scaling and Related Techniques for Geometry Problems}
 Proc.  STOC 1984,  135-143.
\bibitem{JMS04}
J. JaJa, C. W. Mortensen, Q. Shi, {\em Space-Efficient and Fast Algorithms 
for Multidimensional Dominance Reporting and Counting}, Proc.  ISAAC 2004,
 558-568.
\bibitem{KS06}
H. Kaplan, M. Sharir, {\em Randomized Incremental Constructions of Three-dimensional Convex Hulls and Planar Voronoi Diagrams, and Approximate Range Counting},
Proc.  SODA 2006:484-493.
\bibitem{MVY94}
Y. Matias, J.S. Vitter, N. E. Young, {\em 
Approximate Data Structures with Applications}, Proc.  SODA 1994, 187-194.
\bibitem{MNSW98} 
P. B. Miltersen, N. Nisan, S. Safra, A. Wigderson, 
\newblock {\em
 On Data Structures and Asymmetric Communication Complexity}
 J. Comput. Syst. Sci. 57(1): 37-49 (1998).
\bibitem{M06}
C. W. Mortensen,  {\em Fully Dynamic Orthogonal Range Reporting on RAM},
 SIAM J. Comput. 35(6): 1494-1525 (2006).
\bibitem{M06phd}
C. W. Mortensen,
{\em Data Structures for Orthogonal Intersection Searching and Other Problems},
Ph.D. thesis (2006). 
\bibitem{MPP05}
C. W. Mortensen, R. Pagh, M. Patrascu, 
{\em  On Dynamic Range Reporting in One Dimension.} Proc. STOC 2005, 104-111.
\bibitem{N07}
Y. Nekrich, {\em A Data Structure for Multi-Dimensional  Range Reporting}, 
Proc. SoCG 2007, 344-353.
\bibitem{N09arx}
Y. Nekrich, {\em Data Structures for Approximate  Orthogonal Range Counting}, 
arXiv:0906.2738 (2009).
\bibitem{O88}
M. H. Overmars, {\em Efficient Data Structures for Range Searching on a Grid},
 J. Algorithms 9(2): 254-275 (1988).
\bibitem{PD06}
M. Patrascu, E. D. Demaine, {\em 
Logarithmic Lower Bounds in the Cell-Probe Model}, 
 SIAM J. Comput.  35(4):932-963 (2006).
\bibitem{SR95}
S. Subramanian, S. Ramaswamy,  {\em The P-range Tree: A New Data Structure for Range Searching in Secondary Memory}, Proc. SODA 1995,  378-387.
\bibitem{VV96}
D.~ E. Vengroff, J.~S. Vitter,  {\em Efficient 3-D Range Searching in External Memory}, Proc. STOC 1996,  192-201.
\end{thebibliography}
\end{document}